\documentclass[12pt]{article}\usepackage[]{graphicx}\usepackage[]{color}
\makeatletter
\def\maxwidth{ %
  \ifdim\Gin@nat@width>\linewidth
    \linewidth
  \else
    \Gin@nat@width
  \fi
}
\makeatother

\definecolor{fgcolor}{rgb}{0.345, 0.345, 0.345}

\usepackage{framed}
\makeatletter
 {\par\unskip\endMakeFramed%
 \at@end@of@kframe}
\makeatother

\definecolor{shadecolor}{rgb}{.97, .97, .97}
\definecolor{messagecolor}{rgb}{0, 0, 0}
\definecolor{warningcolor}{rgb}{1, 0, 1}
\definecolor{errorcolor}{rgb}{1, 0, 0}
\newenvironment{knitrout}{}{} % an empty environment to be redefined in TeX

\usepackage{alltt}

\usepackage{subcaption}
\usepackage{fullpage}
\usepackage{enumerate}
\usepackage{multirow}
\usepackage{comment}
\usepackage{amsmath,amssymb,amsfonts,amsthm}
\usepackage{bbm}
\usepackage{setspace}
\usepackage{verbatim}
\usepackage{natbib}
\usepackage{bm}
\usepackage{pdflscape}
\usepackage{tikz}
\usepackage{xr}
\usepackage{makecell}

\newcommand{\bx}{\bm{X}}
\newcommand{\ncm}{n_{Cm}}
\newcommand{\ntm}{n_{Tm}}
\newcommand{\bycm}{\bm{y}_{Cm}}
\newcommand{\est}{\hat{\tau}_{\text{rebar}}}
\newcommand{\yci}{y_{Ci}}
\newcommand{\yti}{y_{Ti}}

\newcommand{\yhat}{\hat{y}_C}
\newcommand{\EE}{\mathbb{E}}

\newcommand{\Match}{M}
\newcommand{\match}{m}
\newcommand{\algorithm}{\hat{y}_C(\cdot)}
\newcommand{\covMat}{\bm{X}}
\newcommand{\covVec}{x}

\newtheorem{prop}{Proposition}
\newtheorem{lemma}{Lemma}
\newtheorem{remark}{Remark}

\title{Rebar: Reinforcing a Matching Estimator with Predictions from
  High-Dimensional Covariates}

\author{Adam Sales \& Ben B Hansen \& Brian Rowan}
\IfFileExists{upquote.sty}{\usepackage{upquote}}{}
\begin{document}

\maketitle

\section{Introduction: Two Types of Neglected Data}

Matching-based observational studies in
education sciences often neglect data from the ``remnant'' of a match:
untreated and un-matched subjects.
That is, researchers
will select a set of matched controls that most closely resemble the
treated subjects, and discard data from the remnant, the unmatched
controls.

Similarly, due to sample size and other modeling limitations,
researchers will typically condition their experimental and
observational studies on a small set of pre-treatment covariates that
are deemed most relevant to the study---the variables thought most likely to
pose a confounding threat.
In many cases, reams of less-relevant data are available, perhaps from
state longitudinal data systems or from other sources.
These less relevant covariates are often discarded.

Conducting a causal analysis using only the
matched sample and using only relevant covariates makes good
statistical sense. The data from subjects that are not
part of a match are likely to be distributed
differently than data from the match.
The process of matching encourages researchers to focus their analysis
on the region of common support; the remnant is typically
outside this region by construction.
Including irrelevant variables into an analysis can swamp the sample,
introduce over-fitting or extreme imprecision, and make
impossible common statistical techniques such as ordinary least
squares and logistic regression.

But these excluded data---the remnant and ostensibly irrelevant
covariates---may also contain valuable information.
Perhaps the distribution of the outcome conditional on covariates
could be estimated with more precision by vastly increasing the sample
size using discarded subjects.
Perhaps discarded covariates are not so irrelevant, and capture important baseline differences
between treated and untreated subjects.

This paper is an attempt to thread this needle with a new method that
we call ``remnant-based residualization,'' or ``rebar.''
The idea of rebar is to, on the one hand, extract as
much useful information as possible from the remnant and all available
covariates, and on the other hand to preserve the
most attractive properties of a good matching design.
To implement rebar, we fit a machine learning prediction model to the
unmatched controls---the ``remnant''---predicting
their outcomes in the control condition as a function of the entire
set of covariates.
Using this fitted model, we then generate predicted outcomes for the
matched sample.
Finally, instead of calculating the effect of the treatment on
participants' outcomes themselves, we estimate the intervention's
effect on the difference between participants' predicted outcomes
under the control condition, and their actual outcomes, i.e. their
prediction residuals---this is ``residualization.''
The predictive model need not be correct in any sense, or
consistent or unbiased for any particular parameter.
It must only yield predictions that are closer, on average, to control potential
outcomes than their mean.

Rebar builds thematically on prior work combining matching with
outcome modeling, such as \citet{rubin:1973b} and \citet{ho:etal:2007}, among others, alongside ``doubly robust'' estimation \citep[e.g.][]{kang2007demystifying}.
% the following (til the end of the para) is quoted in full in
% response to rev. 3
Its most direct antecedents are the papers of \citet{rosenbaum2002covariance} and
\citet{biasAdjust}, which suggest forms of residualization
for matching estimators, and of \citet{middleton2011unbiased}, which does
the same for weighting estimators.
Our contribution to that literature is twofold: first, rebar is
\emph{remnant-based}: we argue here that residualization is well
suited to recovering otherwise lost information from the remnant.
Second, we demonstrate by simulation and example how rebar can exploit machine
learning methods and high dimensional covariates without compromising
the classical statistical properties of the match.

Rebar can supplement a wide range of matching
analyses, and may be used alongside other outcome models
and covariate adjustments.

The following section will review causal matching studies, and Section \ref{sec:rebarIntro} will formally introduce rebar.
There, we will discuss a possible threat to the validity of a matching design that rebar can introduce: if the distribution of outcomes, conditional on covariates, differs widely enough between the remnant the matched set, rebar might increase, rather than decrease bias.
We will introduce a diagnostic called ``proximal validation'' that should detect such pathological cases, and suggest ways to tweak the algorithm if a researcher were to confront one.

Rebar can potentially reduce both the bias and the variance of causal estimates, by modeling otherwise unmodeled variation.
That said, this paper will focus its attention on rebar's bias reducing properties.
We will argue, with analytical results (Section \ref{sec:bias}), a simulation study (Section \ref{sec:sim}), and an empirical example (Section \ref{sec:example}) that rebar is an effective method for reducing confounding bias from measured, but unmodeled, confounders in a high-dimensional dataset, without compromising the key advantages of matching.

\section{Matching in Observational Studies: Review}\label{sec:matchingReview}

In an observational study, let $i=1,...,n$ index $n$ subjects, and let $Z_i$ denote subject $i$'s binary treatment assignment, and $Y_i$ subject $i$'s observed outcome of interest.
Assuming non-interference \citep{cox:1958}, and following \citet{neyman:1990} and \citet{rubin:1974}, let $y_{Ti}$ and $y_{Ci}$ denote subject $i$'s (perhaps counterfactual) responses were subject $i$ treated and untreated, respectively.
Then $Y_i=y_{Ti}Z_i+y_{Ci}(1-Z_i)$.
Further, let $\covVec_i$ be a vector of covariates measured prior to treatment.
The potential outcomes $y_C$ and $y_T$ define treatment effects $\tau_i=\yti-\yci$ and a causal estimand
\begin{equation}\label{eq:ett}
\tau_{ETT}=\EE_Z[\bm{\tau}^T\bm{Z}/n_T]=\frac{\bm{\tau}^T\EE\bm{Z}}{n_T},
\end{equation}
the expected average effect of the treatment on the treated.
The expectation in (\ref{eq:ett}) is taken conditional on the posited sampling scheme.

In a matching-based observational study, a researcher will create a
new categorical variable, $\mathbf{\Match}$, considering subjects $i$
and $j$ to be \emph{matched} %``matched''
to one another if $\Match_i=\Match_j$.
(Subjects $i$ with the property that $M_i \neq M_j$ for all $i\neq j$ are \textit{unmatched}.)
Researchers will choose $\Match$ in such a way that matched subjects
have similar covariate distributions $\covVec$.
Perhaps the most popular approach to matching is to use propensity scores \citep{rosenbaum1983central},
$Pr(Z=1|\covVec)$, the probability of
being assigned to treatment conditional on covariates $\covVec$.
In a propensity-score matching design, treated and untreated subjects
are grouped into matches $\Match$ with approximately equal
estimated propensity scores.
Other inexact matching techniques measure subjects' similarity in
$\covVec$ using, for example, Mahalanobis distances \citep{rubin1980bias}
or covariate balance tests \citep{diamond2013genetic}.
Matched sets may contain any (positive) number of treated or untreated
subjects \citep{rosenbaum:1991a}.

Ideally, within any matched set, no subject's \textit{a priori} probability of making its way into the treatment
group was larger or smaller than any other's:
\begin{equation}\label{assumption:experiment}
Pr(Z_i=1| \bm{\Match})=Pr(Z_{j}=1|\bm{\Match})  \text{ whenever } \Match_i=\Match_{j} ;
\end{equation}
this is \textit{perfect matching}.  Under perfect matching in the sense of (\ref{assumption:experiment}), matched comparisons are statistically equivalent to contrasts of treatment and control conditions in block- or paired randomized designs \citep[e.g.,][]{braitman2002roc,rubin2008objective,hansen2011propensity}.
%Let $\ntm$ and $\ncm$ be the number of treated and untreated subjects in match $m$, $\sum_{i: \match_i=m} Z_i$ and $\sum_{i: \match_i=m} (1-Z_i)$, respectively.
%Then for any subject $i$ with $\match_i=m$, $Pr(Z_i=1|\ntm,\ncm)=\ntm/(\ntm+\ncm)$.

A simple matching-based estimator compares average treated and
untreated outcomes within each match.
The average difference between treated and untreated subjects in matched set
$m$ is
\begin{equation*}
t(\bm{Y}_m,\bm{Z}_m)=\frac{\bm{Y}_m^T\bm{Z}_m}{\ntm}-\frac{\bm{Y}_m^T(\bm{1}-\bm{Z}_m)}{\ncm}
\end{equation*}
where $\bm{Y}_m$ and $\bm{Z}_m$ are vectors of $Y$ and $Z$, and $\ntm$ and $\ncm$ are the numbers of treated and untreated among subjects $\{ i : \bm{\Match}_i=m\}$.
Then a matching estimator is
\begin{equation}\label{MatchingEstimator}
\hat{\tau}_\Match (\bm{Y})=\sum_m w_m t_m(Y,Z)
\end{equation}
where weight $w_m=\ntm/n_T$.
Estimator $\hat{\tau}_\Match (\bm{Y})$ is unbiased for $\tau_{ETT}$ under perfect matching
(\ref{assumption:experiment}), or, more generally, if the difference in
assignment probabilities is uncorrelated with control potential
outcomes (Lemma \ref{prop:matchBias} in the appendix).
In practice neither of these will be exactly true, but researchers can hope for approximate unbiasedness, and explore their design's sensitivity to unmeasured (or unmodeled) bias \citep[e.g.][]{gastwirth1998das,hosmanetal2010}.

Frequently, subjects who are not sufficiently
similar in $\covVec$ to other units are left unmatched.
We will refer to the set of unmatched untreated subjects as the
\emph{remnant} %``remnant''
from a match.
Typically, the remnant is discarded.
While discarding data might seem unwise, there is good reason to discard the remnant.
Since no suitable comparisons may be found between subjects in the remnant and treated subjects, any causal comparisons using the remnant necessarily involve modeling $y_C$ as a function of $\covMat$.
Moreover, the remnant typically occupies a mostly separate region of the distribution of $\covMat$ than the matched sample---hence its inability to be matched.
Therefore, comparing outcomes from treated subjects with those from
the remnant involves extrapolation, which can be highly sensitive to
model specification. On the other hand, the remnant may contain information that is useful for modeling $y_C$.

An extensive, occasionally contentious literature discusses variable selection for propensity
score models.
This literature begins with Rubin and Thomas, who advised
erring on the side of inclusiveness, striving to exclude only those
covariates that a consensus of researchers believe to be unrelated to each outcome
variable \citeyearpar[\S~2.3]{rubin:thom:1996}; Rosenbaum's
\citeyearpar[p.76]{rosenbaum:2002} view is similar.  Later contributions
argued that including variables only weakly related to outcomes may
increase the mean squared error (MSE) of effect estimation \citep{brookhart2006variable,austin2011introduction}.  These
additional losses can in principle take the form of bias,
not only variance, even if the MSE-increasing variable was determined in
advance of treatment assignment
\citep{greenland2003qbc,sjolander2009propensity,pearl2009letter}. %,
%although case studies suggest
%these types of bias are often small \citep{liu2012implications,ding2015adjust}.
Most recently,
\citet{steiner2015bias} argued via case study for including all
available covariates, unless ``strong substantive theory'' (p. 573)
suggests the presence of bias-amplifying covariates; ideally,
researchers should include covariates from multiple domains, with each domain including as many
covariates as possible. \citet{pimentel2016constructed} suggested
conducting two analyses, each matching on a different set of covariates.
Methods
attempting to limit the MSE penalty by limiting propensity modeling
variables to those that correlate with observed outcomes have
been met with criticism of a different nature: In Rubin's view, in
order to maximize objectivity, during matching researchers
should keep outcome measurements in a virtual locked box, only to
emerge once the matching structure and other study design elements
have been determined \citep{rubin2008objective}.

Rebar, the method of this paper, is compatible with either attitude to
selection of propensity score variables; our illustration (\S~\ref{sec:example}) emphasizes this
compatibility by adhering to the more restrictive of the two schools.
Without reference to outcome associations, we select for inclusion in
the propensity model those variables we felt that a consensus of
scholars would be most likely to deem potential confounders.  In this
example as in many others, the number of
potential confounders that could  be addressed in this way was
limited: when $p\ge n_T$ or $p\ge n_{C}$, then the
treatment and control samples can ordinarily be separated by
a hyperplane, in the space spanned by
$\bx$, with the result that common binary regression methods
fail to fit \citep{agresti2002categorical,zorn2005solution}; in the
example of \S~\ref{sec:example}, $n_{T}=7$. This heightens the need for additional
measures for confounder control, such as rebar.

\section{Rebar: Using an Outcome Model to Reduce Bias in a Matching Design}\label{sec:rebarIntro}

The procedure we recommend is the following:
\begin{enumerate}
  \item Using the full dataset, construct a match $\bm{\match}$, perhaps
    based on a subset of available covariates, thereby dividing the sample
    into a matched sample and a remnant.
  \item Using units in the remnant, construct an algorithm
    $\algorithm$ to predict $y_C$ as a function of the full matrix $\covMat$.
  \item Assess the performance of $\algorithm$ (See Section \ref{sec:remnant})
  \item For all subjects  $i$ in the matched sample, use $\algorithm$
    to predict $y_{Ci}$ as $\hat{y}_{Ci}=\hat{y}_{C}(\covVec_i)$.
  \item Construct prediction errors $e\equiv Y-\hat{y}_C(\covMat)$ for all
    subjects in the matched sample.
  \item Estimate treatment effects in the matched sample,
    substituting $e$ for $Y$ in the outcome analysis.
\end{enumerate}
As in \citet{rosenbaum2002covariance}, the model $\algorithm$ relating $\covMat$ and $y_C$ is an algorithmic model, rather than a statistical model.
That is, it does not estimate parameters of a probability distribution, but rather generates deterministic predictions of $y_C$ when given a vector $\covVec$.
Since this procedure relies on the residuals of a model fit to $Y$, we
will refer to it as ``residualization.''

The predictions $\yhat (\covVec)$ bear some similarity to prognostic
scores \citep{hansen:2008biometrika}.
Prognostic scores, which are analogous to propensity scores, are
statistics that are sufficient for the relationship between $y_C$ and
$\covVec$.
They are commonly understood as predictions of $y_C$ as a function of
$\covVec$ \citep[e.g.][]{pane2013effectiveness}.
In fact, much of the intuition behind prognostic scores supports our
use of $\yhat (\covVec)$ here, though the prognostic score theory will not
play a direct role in our argument.

Now as above, define residuals
\begin{equation*}
e=Y-\yhat (\covVec).
\end{equation*}
Then we may define ``potential residuals'': $e_C=y_C-\yhat (\covVec)$ and $e_T=y_T-\yhat (\covVec)$.
Analogously to $Y$, the observed residuals are $e=Ze_T+(1-Z)e_C$.
Crucially,
\begin{equation}\label{eq:eDiff}
e_{Ti}-e_{Ci}=\tau_i ,
\end{equation}
where $\tau_i$ as above is subject $i$'s treatment effect, $\yti-\yci$.
To see this, note that $y_C=\hat{y}_C(\covMat)+e_C$ and $y_T=\hat{y}_C(\covMat)+e_T=\hat{y}_C(\covMat)+e_C+\tau$.
The prediction $\yhat (\covVec)$ is based only on pre-treatment variables
$\covVec$, and not on treatment status $Z$ from subjects in the matched sample.
That being the case, it cannot be affected by treatment status---we would counterfactually estimate the same $\yhat (\covVec)$ for alternative realizations of $\bm{Z}$ in the matched set.
Therefore, we can write $e_{Ti}-e_{Ci}=y_{Ti}-\hat{y}_{Ci}-(y_{Ti}-\hat{y}_{Ci})=y_{Ti}-y_{Ci}=\tau_i$: the treatment effect is manifest entirely in the residuals $e_C$ and $e_T$, and not at all in $\yhat (\covVec)$.

The prediction errors $e$, then, may replace $Y$ in an outcome analysis.
In particular, replace matched-set-specific treatment-control differences in $Y$, $t_m(Y,Z)$ with differences in $e$: $t_m(e,Z)$.
That is, let
\begin{equation*}
t_m(e,Z)=\bar{e}_{m, Z=1}-\bar{e}_{m, Z=0}=\frac{1}{n_{Tm}}\sum_{i:\match_i=m} e_i Z_i -\frac{1}{n_{Cm}}\sum_{i:\match_i=m} e_i (1-Z_i)
\end{equation*}
then define
\begin{equation}\label{estimator}
\est=\sum_m w_m t_m(e,Z)
\end{equation}
Residualization, then, means revising a
matching estimator by replacing outcomes $Y$ with observed value/$\algorithm$ differences; it aims to rid the
dependent variable of variation that is not informative about
treatment effects.
% the following 2 sentences are quoted in full in
% response to rev. 3
Rosenbaum \citeyearpar{rosenbaum2002covariance} precedes conventional hypothesis tests
with a residualization step, using observations within the matched sample
to fit the prediction model. If one instead trains one's
prediction algorithm $\algorithm$ using the remnant of the matching
procedure, the method becomes compatible with common estimation
(as well as hypothesis testing) techniques, and may offer larger
numbers of observations for training $\algorithm$.
Such \textit{remnant-based} residualization, briefly ``rebar,'' is the
topic of this paper.

%% Rebar can also be short for
%% ``reinforcement bar,'' a metal beam often embedded in concrete, which reinforces buildings from against a variety of threats.
%% We will argue that, analogously, rebar can reinforce the conclusions from a matching design.

%% Specifically, we will argue that using rebar to reinforce a matching design can reduce confounding bias; there is good reason to believe that it will reduce the sampling variability of matching estimators as well---as we will see in Sections \ref{sec:sim} and \ref{sec:example}---but we focus here on bias.
%% Briefly, the argument is that, by subtracting predictions $\yhat (\covVec)$ from outcomes Y, researchers free the estimate from variation in $Y$ that is not informative about treatment effects.
%% If this variation is correlated with treatment assignment $Z$, then rebar will reduce confounding bias.
%% These ideas will be made more precise in Section \ref{sec:bias}.

\subsection{Cross Validation and Proximal Validation: Assessing  $\algorithm$}\label{sec:remnant}

Using the remnant to model outcomes as a function of covariates affords the researcher a great deal of flexibility.
Researchers may use data from the remnant---both covariates and
outcomes---to attempt a variety of prediction techniques, and choose
the one which performs best.
This is particularly important when the dimension of $\covMat$ is large, so formulating statistical models based on theory or first principles is hard or impossible; a variety of methods must be attempted.
A useful tool in this regard is $k$-fold cross-validation
\citep{crossvalidation}, which can estimate the predictive accuracy of
a model using data from the training sample.
Cross-validation results may be examined for bias, variance, or other
measures of predictive performance, but Proposition~\ref{biasBound}
(below) suggests a
focus on prediction mean-squared-error.
In the rebar case, cross validation using data from the remnant can estimate $MSE_{\text{remnant}}=\EE_{i\in remnant} (\hat{y}_{Ci}-y_{Ci})^2$ or $R^2_{\text{remnant}}=1-MSE_{\text{remnant}}/Var_{\text{remnant}}(y_C)$.\footnote{In defining $MSE_{\text{remnant}}$ and $R^2_{\text{remnant}}$ thusly, we briefly depart from our convention of conditioning on potential outcomes and instead treat them as random, drawn from the same superpopulation as the remnant. $MSE_{\text{remnant}}$ and $R^2_{\text{remnant}}$ do not play a role in the theoretical development of rebar, but are useful heuristics in practice.}
These results can be used both to pick a modeling technique and to
pick tuning parameters.
After modeling choices have been made, researchers arrive at an
estimated prediction function $\algorithm:\mathbb{R}^p\rightarrow
\mathbb{R}$ that generates predictions $\hat{y}_C(\covMat)$ as a function of
covariates $\covMat$.

Cross-validation estimates an algorithm's predictive performance when
applied to new cases drawn from the same population as the training
set.
Of course, this is manifestly not the case for rebar.
Subjects in the matched sample are likely to be different from those
in the remnant; a model fit and cross-validated in the remnant may not perform as well in
the matched sample as that validation would suggest.
Write $\mathcal{S}_{\Match}$ to denote the matched sample, i.e. $\{i :
\exists j \neq i \text{ s.t. } M_{i} = M_{j}\}$. One expects $\mathrm{MSE}_\mathrm{remnant}$ to be less than
$\mathrm{MSE}_{\Match}=\{\sum_{i \in \mathcal{S}_{\Match}}
(\hat{y}_{Ci}-y_{Ci})^2\}/|\mathcal{S}_{\Match}|$, and  $R^2_\mathrm{remnant}$
to be less than $R^2_{\Match}$.
This is unfortunate but far from fatal---the more information a
prediction algorithm can learn about the matched sample from the
remnant the better rebar can reinforce a causal design.
Perfection is not necessary.

One does not expect $\mathrm{MSE}_{\Match}$ to exceed $\{\sum_{i \in \mathcal{S}_{\Match}}
(y_{Ci} -\overline{{y}_{C}}_{\mathcal{S}_{\Match}})^2\}/|\mathcal{S}_{\Match}|$,
although this can occur.
In such cases rebar could do more harm than good. Even with perfect
matching in the sense of~\eqref{assumption:experiment}, it could
diminish efficiency; and if~\eqref{assumption:experiment} is only
approximately true, rebar could increase bias as well.

Fortunately, simple diagnostic tools can identify such pathological cases.
Further, in many of \emph{those} cases there are simple modifications
to rebar that will improve its performance.
To illustrate a diagnostic that we call ``proximal validation,''
consider full matching within calipers of width $c_0$ in terms of a
continuous variable or index, such as the propensity score.
All control subjects within $c_0$ of a treated
subject are matched, with remaining controls constituting the remnant.
How well does an algorithm $\algorithm$ fit in the remnant perform in the matched sample?
To gauge $\algorithm$'s performance, a researcher will subdivide the remnant into two groups by using caliper $c_1>c_0$ to construct a new, larger matched set.
The cases in the remnant that are matched under the more permissive caliper $c_1$ are ``proximal'' cases---whether they are matched depends on the choice of caliper.
The cases that remain unmatched even under $c_1$ are ``distal'' cases, unmatchable under either scheme.
Proximal validation re-fits $\algorithm$ using only data from
subjects in the distal remnant, then examines its performance on the proximal portion of the remnant.
If $\algorithm$ performs poorly when extrapolated from the remnant to the matched set, it likely also performs poorly when extrapolated from distal cases to proximal cases within the remnant.
In other words, proximal validation is a way to gauge the performance of $\algorithm$ when its results are extrapolated in a way analogous to a matching design.

As compared to estimating $\mathrm{MSE}_{\Match}$ with rebar's MSE on
the control group, proximal validation permits the analyst to keep matched
subjects' outcomes in Rubin's \citeyearpar{rubin2008objective} virtual
locked box, even as the rebar model is being validated and improved.
Proximal validation is not limited to propensity-score full-matching designs with calipers; it may be used with any matching design that involves a quantitative restriction on allowable matches.
The procedure, in general, will be to slightly relax that restriction, choose a second, more expansive match, and use the results to divide the remnant into proximal and distal portions.

If $\algorithm$'s performance in proximal validation is discernibly
worse than its cross-validation performance, the rebar routine should be modified.
Suppose the mechanism selecting untreated units between the remnant and
the matched sample is matching based on an estimated propensity score.
In this case, the estimated propensity score itself can be incorporated into
the prediction model $\algorithm$---for instance, by including interaction terms
between the columns of $\covMat$ and $\hat{\pi}$.

Another useful diagnostic test is to check covariate balance on the predictions $\bm{\hat{y}}_C(\covMat)$.
Since $\hat{y}_C(\covMat)$ is a covariate, a successful matching
design will ensure that its distributions are similar among treated
and matched untreated subjects.  Even though $\hat{y}_C(\covMat)$ is a
constructed variable, its balance can be tested in the
same ways as balance on manifest variables, since the model behind it is fit without
reference to the matched sample.
If a balance test rejects the hypothesis of $\hat{y}_C(\covMat)$ balance, researchers may revise either the prediction algorithm $\algorithm$, the matching scheme, or both.

\section{Rebar's Effects on Bias}\label{sec:bias}
To see the potential of rebar to reduce the bias of a matching estimator, note that
the rebar estimator $\est$ can be expressed as the difference between two estimated treatment effects:
\begin{equation}\label{eq:rebarDiff}
\est = \hat{\tau}_\Match (\bm{Y})-\hat{\tau}_\Match (\bm{\hat{y}}_C)
\end{equation}
the matching estimator of the effect of the treatment on $Y$, minus an estimate of the effect of the treatment on $\hat{y}_C(\covMat)$.
To see this, note that:
  \begin{align*}
t_m(e,Z)=&\frac{1}{n_{Tm}}\sum_{i:\bm{\Match}_i=\match} e_i Z_i
-\frac{1}{n_{Cm}}\sum_{i:\bm{\Match}_i=\match} e_i (1-Z_i)\\
=&\left(\frac{1}{n_{Tm}}\sum_{i:\bm{\Match}_i=\match} Y_i Z_i
         -\frac{1}{n_{Cm}}\sum_{i:\bm{\Match}_i=\match} Y_i (1-Z_i)\right)-\\
&\left(\frac{1}{n_{Tm}}\sum_{i:\bm{\Match}_i=\match} \hat{y}_{Ci} Z_i
        -\frac{1}{n_{Cm}}\sum_{i:\bm{\Match}_i=\match} \hat{y}_{Ci} (1-Z_i)\right)\\
\equiv& \Delta Y_m -\Delta \hat{y}_{Cm}.
\end{align*}
The expression in (\ref{eq:rebarDiff}) follows by taking weighted averages of $\Delta Y_m$ and $\Delta \hat{y}_{Cm}$.
Of course, the treatment cannot have an effect on $\hat{y}_C(\covMat)$, which is a function of pre-treatment covariates and a separate sample; any observed ``effect'' of the treatment on $\hat{y}_C(\covMat)$ must be the result of covariate imbalance.

Two properties of the rebar estimate follow immediately.
%First,
\begin{prop}\label{prop:biasSum}
\begin{equation*}
bias(\est)=bias(\hat{\tau}_\Match (\bm{Y}))- \hat{\tau}_\Match (\bm{\hat{y}}_C)
\end{equation*}
\end{prop}
Viewing $\hat{\tau}_\Match (\bm{\hat{y}}_C)$ as an
estimate of $\est$'s bias, the effect of
residualization is to subtract from the matching estimator an estimate
of its bias.  (As with other bias correction methods, it backfires when the bias is poorly
estimated, an eventuality proximal validation aims to detect.)

%Next,
\begin{prop}\label{prop:noBias}
Under perfect matching (\ref{assumption:experiment}), $\est$ is unbiased for $\tau_{ETT}$.
\end{prop}
This follows since, when treatment is essentially randomized within matches, $\EE \hat{\tau}_\Match (\bm{Y})=\tau_{ETT}$ and $\EE\hat{\tau}_\Match (\bm{\hat{y}}_C)=0$.
So in a successful matching design, rebar does not introduce bias.
Propositions \ref{prop:biasSum} and \ref{prop:noBias} hold for any
effect estimator $\hat{\tau} (\cdot)$ that is linear in outcomes $Y$,
i.e. for which (\ref{eq:rebarDiff}) holds.

\subsection{An Upper Bound on the Bias of the Rebar Estimator}\label{sec:upperBound}

The closer, on average, predictions $\hat{y}(\covVec)$ are to control potential outcomes in the matched set, the smaller the bias of $\est$ must be.

\begin{prop}\label{biasBound}
In a matching design, the squared bias of $\est$ can be bounded as

\begin{equation*}
bias(\est)^2\le MSE_{\Match}\times C(n,\bm{n_T},\bm{n_C})
\end{equation*}
where $MSE_{\Match}=\sum_{i\in \text{matched}} (\hat{y}_{Ci}-y_{Ci})^2/n_{\Match}$,   $n_{\Match}$ is the number of subjects in the matched set, and
\begin{equation*}
C(n,\bm{n_T},\bm{n_C})=\frac{n}{n_T^2}\sum_m
  (\ncm+\ntm) \max\left(1,\frac{\ntm}{\ncm}\right)^2.
\end{equation*}
Equivalently,
\begin{equation*}
\left(\frac{bias(\est)}{SD(y_C)}\right)^2\le(1-R^2_{\Match})\times C(n,\bm{n_T},\bm{n_C})
\end{equation*}
Where $SD(y_C)$ is the sample standard deviation of $y_C$ in the matched set and $R^2_{\Match}$ is the prediction $R^2$ in the matched set, $1-\sum_{i\in \text{matched}} (y_{Ci}-\hat{y}_{Ci})^2/\sum_{i\in \text{matched}} (y_{Ci}-\bar{y}_{Cmatched})^2$.
\end{prop}

A proof of proposition~\ref{biasBound} appears in the Appendix.

\begin{remark}\label{biasBoundRemark}
In a pair-matching design $C(n,\bm{n_T},\bm{n_C})=4$.

\end{remark}

Therefore, the bias of $\est$ can be bounded as a function of the
average squared error of the prediction algorithm in the matched set.
Were it possible to perfectly predict all subjects'
$y_C$
values, their treatment effects could be estimated unbiasedly
(exactly, in fact).
More broadly, Proposition \ref{biasBound} suggests that prediction
algorithms need not be based on a correct model to yield
estimates with low bias.
They must merely be accurate, on average.
This, in turn, suggests that machine learning algorithms, whose
central purpose tends to be prediction, can serve well as
residualization mechanisms.

In practice, the bounds in Proposition \ref{biasBound} are unobservable, since they involve control potential outcomes in the matched set, which are only observable for the matched controls.
Further, since the prediction algorithm $\algorithm$ is fit in the remnant, the bounds are not directly estimable without strong assumptions.
But based on cross-validation estimates of $MSE_{\text{remnant}}$ and $R^2_{\text{remnant}}$, and an assessment of $\algorithm$'s sensitivity to extrapolation from proximal validation, researchers can formulate reasonable guesses as to the values of $MSE_{\Match}$ and $R^2_{\Match}$.

Proposition \ref{biasBound} assumes nothing about subjects' respective
probabilities of treatment assignment within matches.
In particular, it allows for a situation in which some subjects may
be assigned to treatment with probability 1---this is a rather extreme
violation of the stratified randomization assumption (\ref{assumption:experiment}).
Under weak assumptions about the distribution of treatment
assignments, the bound in Proposition~\ref{biasBound} may be
considerably tightened.
For instance, \citet{rosenbaum:2002} suggests a general model for sensitivity
analysis for observational studies: the assumption that for some $\Gamma\ge
1$, if
$\match_i=\match_j$---that is, $i$ and $j$ are in the same matched
set---and $P_i=Pr(Z_i=1)$ and $P_j=Pr(Z_j=1)$, then
\begin{equation}\label{eq:rbounds}
  \frac{1}{\Gamma}\le\frac{P_i(1-P_j)}{P_j(1-P_i)}\le \Gamma.
\end{equation}
That is, for matched subjects $i$ and $j$, the ratio of the odds that
$i$ is selected for treatment to the odds that $j$ is selected is
bounded by $1/\Gamma$ and $\Gamma$.
The following proposition uses the framework in (\ref{eq:rbounds}) to
tighten the bound in Proposition~\ref{biasBound} in the simple case of a
matched-pair design; an analogous result may hold for more
complex designs, but we leave such an extension for future work.

\begin{prop}\label{prop:gammaBound}
  In a pair-matching design, if (\ref{eq:rbounds}) holds for some
  $\Gamma\ge 1$, then
  \begin{equation*}
    bias(\est)^2\le
    MSE_{\Match}\times 4\left(\frac{\Gamma^{1/2}-1}{\Gamma^{1/2}+1}\right)
  \end{equation*}
  Equivalently,
  \begin{equation*}
   \left(\frac{bias(\est)}{SD(y_C)}\right)^2\le(1-R^2_{\Match})\times 4\left(\frac{\Gamma^{1/2}-1}{\Gamma^{1/2}+1}\right)
   \end{equation*}
\end{prop}

A proof of Proposition~\ref{prop:gammaBound} is given in the Appendix.
\begin{remark}
  For $\Gamma=6$, which \citet[][p. 114]{rosenbaum:2002} characterized
  as ``a high degree of insensitivity to hidden bias,''
  $4\left(\frac{\Gamma^{1/2}-1}{\Gamma^{1/2}+1}\right)\approx 1.7.$
  That is, a very weak assumption about the balance of treatment
  assignment probabilities in a matched pair design constricts the
  bound in Proposition~\ref{biasBound} by more than half. If
  $\Gamma=3$, the multiplier on $(1-R^2_{\Match})$ is
  approximately one. On the other hand, as $\Gamma\rightarrow \infty$,
  the multiplier approaches 4, as in Remark \ref{biasBoundRemark}.
\end{remark}

Propositions \ref{biasBound} and \ref{prop:gammaBound} show that by using data from the remnant and covariate matrix $\covMat$ to predict potential outcomes $y_C$, researchers can substantially bound the the bias of their treatment effect estimates.
The closer the estimates are to the true values, on average, the lower the bound on the bias---the algorithm $\algorithm$ need not be correct in any sense, only predictive.

\section{A Simulation Study}\label{sec:sim}

This section presents a simulation study with two principal goals: to
demonstrate rebar's potential to improve upon matching
estimators under a variety of circumstances, and rebar's ability to
interact with, and improve upon, a variety of matching designs and
estimators.
A second, smaller study examines rebar's performance under
pathological circumstances.

\subsection{Data Generating Models}\label{sec:sim-data-generating}
The study imagined a researcher estimating the effect of a treatment
$Z$ on an outcome $Y$, using a sample of $n=$400 subjects, in
the presence of $p=$600 covariates.
While all of the covariates were potential confounders, the simulated researcher knew that five of the covariates---the first
five columns of covariate matrix $\covMat$---predict both $y_C$ and
$Z$; prior background knowledge provided little guidance regarding the
remaining 595.

The outcomes $y_C$ were generated as a linear function of a
fixed covariate vector $\covVec_i$:
\begin{equation}\label{eq:outcome}
y_{Ci}=\bm{1 }^T x_{i, 1:5}+\bm{\beta }^T x_{i,6:600}+\epsilon_i
\end{equation}
where the coefficients $\bm{\beta}$ were drawn from an exponential
distribution with a rate of $\lambda=5$ and $\epsilon$ was drawn from a standard normal distribution.
 A ``treated'' group was selected
 according to probabilities
\begin{equation}\label{eq:treatment}
Pr(Z_i=1|\covVec_i)=\mathrm{logit}^{-1}(\alpha^*+\bm{1 }^Tx_{i,
  1:5}+\kappa\bm{\beta }^Tx_{i,6:600}).
\end{equation}
That is, the log odds of treatment assignment were linear in covariates.
We chose the parameter $\alpha^*$ in such a way that, on average,
$n_T=$50 were treated.
As in (\ref{eq:outcome}), the coefficients for the first five columns of $\covMat$ in (\ref{eq:treatment}) were all set equal to 1.
The coefficients of the other 595 columns in  (\ref{eq:treatment}) were the same as in (\ref{eq:outcome}), multiplied by a factor $\kappa$ which varied between simulation runs.

The factor $\kappa$ controlled the amount of confounding after matching.
When $\kappa=0$, only the first five columns of $\covMat$ predicted $Z$,
so estimates from a match based on those covariates were approximately unconfounded.
When $\kappa>0$, every column of $\covMat$ predicted both $Z$ and $y_C$, and therefore confounded matching estimators that used only the first five columns of $\covMat$.
As $\kappa$ increased, so did the magnitude of the bias due to
confounding after the match; the three values we assigned,
$\kappa=$0, 0.1, 0.5, roughly correspond
to zero, low, and high unmatched confounding.

A second parameter, $\rho$, controlled the covariance structure of
$\covMat$, effectively controlling the ease of predicting $y_C$ as a
function of $\covMat$.
In this simulation, $\rho=$0, 0.004, and 0.05.
The rows of $\covMat$ were generated from a $p=$600-dimensional
multivariate normal distribution, with a random
covariance matrix whose eigenvalues we
specified (it was generated with \verb|R| code of \citealt{randomMatrix}).
We set these eigenvalues $ev_k$, $k=1,...,$600, to decay exponentially: $ev_k=\exp\{-\rho k\}$.
When $\rho=0$, all eigenvalues were unity, and the columns of $\covMat$ are
uncorrelated.
As $\rho$ increased, the columns of $\covMat$ became increasingly
correlated: there was low-dimensional structure in $\covMat$.
Prediction algorithms typically perform better when high-dimensional $\covMat$ can be summarized with a low-dimensional structure.
During the simulation we recorded the estimated
prediction $R^2$ from the cross-validation, and models fit to $\covMat$
with higher $\rho$ fit substantially better.

Covariates $\covMat$ and coefficients $\bm{\beta}$ varied between
scenarios (one random matrix $\covVec$ for each value of $\rho$ and one
random vector $\bm{\beta}$ for each value of $\kappa$)
but were held fixed across simulation runs within scenarios.
Outcomes $Y$ and treatment assignments $Z$ were generated anew in each
simulation run.
Each run, all ten effect estimates were computed using the same data.

\subsection{Treatment-Effect Estimators}\label{sec:sim-estimators}

In each simulation run, we constructed four matches.
Each of these matches, in turn, gave rise to two or three
treatment-effect estimates; all in all, we compared 10 different estimators.
These are summarized in Table \ref{tab:estimators}.

\begin{table}
  \begin{tabular}{ccc}
    \textbf{Matching Method}&\textbf{ Matching Variables} &\textbf{ Adjustment Method(s)}\\
    \hline
    \hline
    Optimal pair matching & $\covMat_{1:5}$& Rebar\\
    \hline
    Nearest Neighbor & $\covMat_{1:5}$& \makecell{Bias adjusted\\
    bias adjusted+rebar}\\
    \hline
    Coarsened Exact Matching & $\covMat_{1:5}$ &\makecell{Within-sample
                               OLS\\within-sample OLS+rebar}\\
    \hline
    High-dimensional pair match & $\covMat$ & Rebar\\
    \hline
 \end{tabular}
 \caption{Summary of the matching and estimation methods in the
   simulation study.}
 \label{tab:estimators}
\end{table}

%\subsubsection{Optimal Propensity Score Pair Matching}
\textbf{Optimal Propensity Score Pair Matching}\\
We estimated propensity scores using logistic regression, with $\bf{Z}$ regressed on the matching covariates, the first five columns of $\covMat$.
Using these propensity scores, we constructed an optimal pair match
without replacement---each treated subject was matched to a unique
control subject in such a way that the total distance in propensity
scores between matched subjects was minimized. (We used the \texttt{optmatch} package in \texttt{R}
[\citealt{hansen2007optmatch}] and chose pair matching strictly for
ease of interpretation; %, not because it is the best or most easily generated without-replacement matching structure; for instance, the
the application of \S~\ref{sec:example} uses \texttt{optmatch} to pair
each treated subject to 1--4 controls.)
We first estimated treatment effects via (\ref{MatchingEstimator}), the
average difference in $Y$ between treated subjects and their matched
controls, without adjustment from an outcome model.

Next, we computed rebar-adjusted estimates.
With the remnant from the pair-match as a training set,  we used a
combination of lasso \citep{tibshirani1996regression} and random forests
\citep{breiman2001random} to construct $\algorithm$, a predictor of control
potential outcomes $y_C$ as a function of the entire covariate matrix $\covMat$.
We implemented these in \texttt{R} with the
\texttt{glmnet} and \texttt{randomForest} packages
\citep{glmnet,rfCite}, and tuned and combined them with the
\texttt{SuperLearner} package \citep{superlearner} to minimize
mean-squared-error.
As outlined in Section~\ref{sec:rebarIntro}, we used the fitted
$\algorithm$ to construct predictions $\yhat$ and prediction errors
$e$ in the matched set, and estimated the treatment effect as in
equation (\ref{estimator}).

\textbf{Nearest-Neighbor Propensity Score Matching}\\
%\subsubsection{Nearest-Neighbor Propensity Score Matching}
Using the same propensity scores as in the optimal pair match,
we constructed a ``nearest-neighbor'' match, as proposed by \citet{nearestNeighbor},
and implemented by the \texttt{Matching} package in \texttt{R} \citep{matching}.
We used the ``ATT'' estimator of \citet{nearestNeighbor} to estimate
the average of the differences between each treated subject's outcome
and the average outcome of its matched controls.
Next, we computed the ``bias adjusted'' estimator suggested in
\citet{biasAdjust}, using an ordinary least squares (OLS) outcome model fit
to the matched sample.\footnote{\citet{biasAdjust} in fact suggest a more
  complicated regression routine that includes non-linear terms and
  interactions as the sample size grows, but in practice implement the
  routine with OLS; the \texttt{Matching} package similarly uses OLS.}
Since OLS cannot be fit when the number of covariates exceeds the
sample size, we used only the matching covariates for the bias adjustment.
Finally, we combined this within-sample bias adjustment with rebar.
As in optimal pair matching, we fit the lasso/random
forest/SuperLearner algorithm to data from the remnant
of the nearest neighbor match, predicting $y_C$ as a function of the entire matrix
$\covMat$, and computed $\yhat$ and $e$ in the matched set.
To estimate effects with both within-sample and rebar adjustment, we
substituted $e$ for $Y$ in the bias-adjusted estimator.

%\subsubsection{Coarsened Exact Matching}
\textbf{Coarsened Exact Matching}\\
We constructed a coarsened exact match, as described
in \citet{cem} and implemented in \texttt{R} with the \texttt{cem}
package \citep{cemR}.
We coarsened each of the first five columns of $\covMat$ with five
bins, matched exactly on the coarsened covariates, and estimated
treatment effects via (\ref{MatchingEstimator}).
Next, we constructed a within-sample adjusted estimator along the
lines of \citet{ho2007matching}: using only data from the matched
sample, we regressed $Y$ on $Z$ and the first five
columns of $\covMat$, and recorded the coefficient on $Z$.
Finally, we combined the within-sample adjustment with rebar.
As in the optimal pair and nearest neighbor analyses,
we used data from the remnant to fit a lasso/random
forest/SuperLearner algorithm predicting $y_C$ as
a function of the entire $\covMat$, and generated predictions $\yhat$
and errors $e$ in the matched set.
To estimate effects, we regressed $e$ on $Z$ and the first five
columns of $X$, and recorded the coefficient for $Z$.

%\subsubsection{High-Dimensional Pair Match}
\textbf{High-Dimensional Pair Match}\\
The first three matching designs, optimal pair matching, nearest
neighbor matching, and coarsened exact matching, used only the first
five columns of $\covMat$---the known confounders.
However, when presented with a set of $p=600$ covariates, many
real-world researchers would not stop at the first five.
Instead, they would try to incorporate additional covariates into their matches.
The resulting iterative process of matching and balance checking is
difficult or impossible to simulate; however, there are
a number of automatic machine learning algorithms for estimating
probabilities in high-dimensional spaces \citep[e.g.][]{twang,highDimPS}.
In this vein, in parallel to the rebar prediction model $\algorithm$,
we estimated high-dimensional propensity scores with random forest
classification and lasso logistic regression, tuned and combined via the
SuperLearner.
We used these high-dimensional propensity scores to construct a second
optimal pair match.
As in the conventional pair match, we estimated effects using equation
(\ref{MatchingEstimator}) and, fitting algorithm $\algorithm$ to the
remnant, we computed a rebar estimate.

\subsection{Simulation Results}\label{sec:sim-results}

%\begin{landscape}

\begin{figure}
\centering
\begin{knitrout}
\definecolor{shadecolor}{rgb}{0.969, 0.969, 0.969}\color{fgcolor}
\includegraphics[width=\maxwidth]{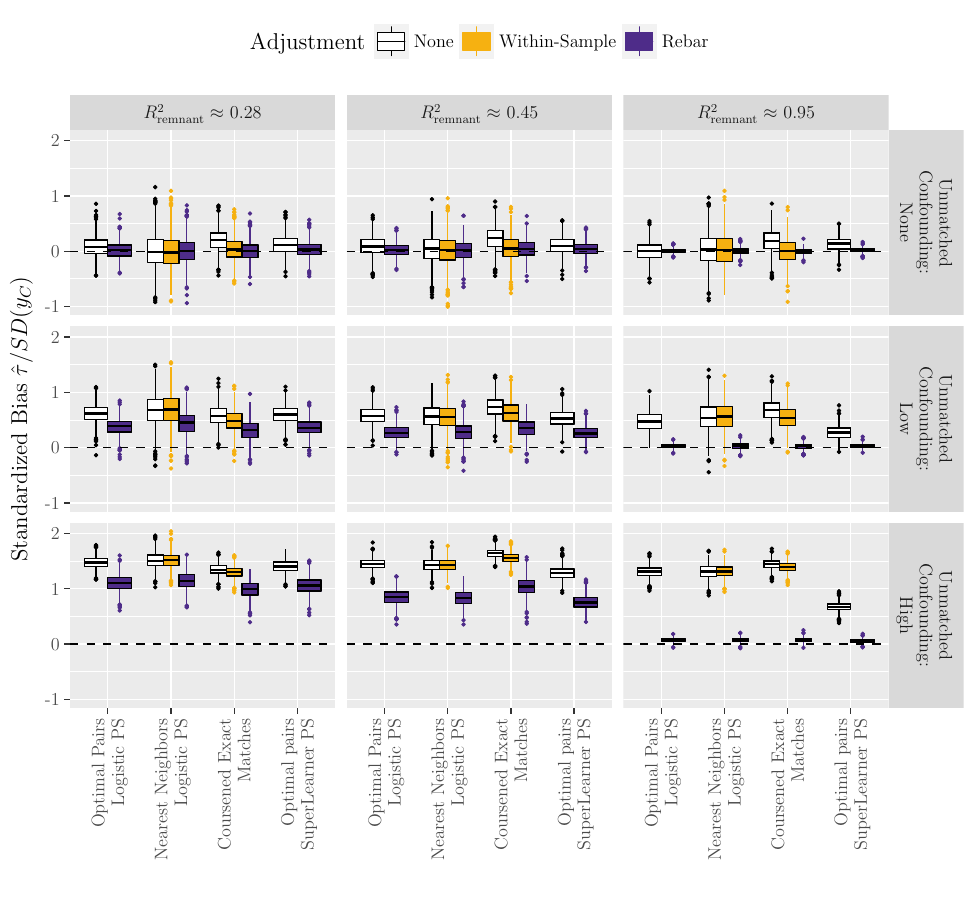} 

\end{knitrout}
\caption{Boxplots of treatment effect estimates from 1000
  simulation runs under the data generating models in Section
  \ref{sec:sim-data-generating}. The true treatment effect of zero is indicated by a
  horizontal dotted line. The estimated treatment effects were
  divided by the standard deviation of $y_C$. % to ease interpretation.
  The matching and outcome adjustment methods are described in Section
  \ref{sec:sim-estimators} and Table \ref{tab:estimators}; the rebar
  adjustments to the nearest neighbor and coarsened exact match were
  done alongside within-sample adjustment.%, optimal pair
  %matching, nearest-neighbor matching, coarsened exact
  %matching, and optimal pair matching using machine-learning
  %propensity scores, were either unadjusted or adjusted with rebar
  %alone, with within-sample adjustment, or both.
The nine simulation scenarios, described in
Section~\ref{sec:sim-data-generating}, are arranged in a matrix, with
rows for $\kappa=$0, 0.1, and 0.5,
% denoting levels of confounding after the match,
and columns for $\rho=$0, 0.004, and
0.05. %, from left to right, denoting correlation structure
                 %between covariates.
The $R^2_{\text{remnant}}$ values listed are averages of prediction $R^2$ for $\algorithm$ estimated using cross-validation within the remnant.}
\label{fig:simulationResults}
\end{figure}
%\end{landscape}

Figure \ref{fig:simulationResults} shows the results of the simulation, after 1000 simulation runs.
Each row of Figure \ref{fig:simulationResults} corresponds to a value of $\kappa$; in the first row, $\kappa=$0, corresponding to no confounding from the covariates not used in the match, in the second row $\kappa=$0.1, corresponding to moderate confounding from the left-out covariates, and in the third row $\kappa=$0.5, corresponding to a high degree of confounding.
Each column of Figure \ref{fig:simulationResults} corresponds to a different value of $\rho$: 0, 0.004, and 0.05.
These correspond to datasets increasingly amenable to prediction
algorithms; the top of the figure lists the average cross-validation
$R^2_{\text{remnant}}$ of $\algorithm$ fit in the remnants from the
pair matches.
Each panel of Figure \ref{fig:simulationResults} displays boxplots of the
ten treatment effect estimates, divided by the standard deviation of
$y_C$.
% : four matching estimates, from optimal pair, nearest-neighbor,
%  coarsened exact, and high-dimensional pair matching. Nearest Neighbor and coarsened exact
% matching estimates are shown alongside estimates incorporating within-sample bias adjustments.
% Rebar estimates are shown for each matching design: for optimal and
% high-dimensional pair
% matching, rebar is the only outcome model adjustment; for
% nearest-neighbor and coarsened exact matching, rebar and within sample adjustments are combined.

A number of patterns are apparent.
When $\kappa=0$, the covariates not used in the match did not pose a
confounding threat, and all the estimators % (with the slight exception of
%coarsened exact matching)
were approximately unbiased.
Rebar reduced the variance of the effect estimates, subtly for the first two columns and dramatically in the third.
As $\kappa$ increased, all effect estimates became increasingly biased.
However, rebar substantially reduced the bias.
Rebar was similarly effective when used on its own and when used in conjunction with within-sample outcome model adjustments---that is, rebar had quite a bit to add even after other adjustments.
Unsurprisingly, rebar's performance, both in terms of bias and
variance reduction, improved with higher $R^2_{\text{remnant}}$---the closer,
on average, the predictions $\hat{y}_C(\covMat)$ are to $y_C$ in the
remnant (and, presumably, in the matched set, too), the more good
rebar can do.

The high-dimensional propensity score match demonstrated that rebar
can improve upon designs that incorporate all of $\covMat$.
% In fact, in most cases the high-dimensional propensity score estimators performed worse than traditional estimators.
% While investigations into the causes of this behavior are beyond the
% scope of this paper, it is worth noting that matching diagnostics
% showed the high-dimensional propensity score model to be severely overfit.
% In any case, rebar substantially improved these results, lessening
% their bias to an extent similar to the other estimators.

This simulation study showed rebar's potential: rebar can substantially
reduce both the bias and the variance of a matching estimator,
especially in the presence of high-dimensional confounding and with an
accurate prediction algorithm.

\subsection{Rebar's Performance Under Non-Linearity}\label{sec:sim-non-linear}
\begin{figure}
\centering
\begin{knitrout}
\definecolor{shadecolor}{rgb}{0.969, 0.969, 0.969}\color{fgcolor}
\includegraphics[width=\maxwidth]{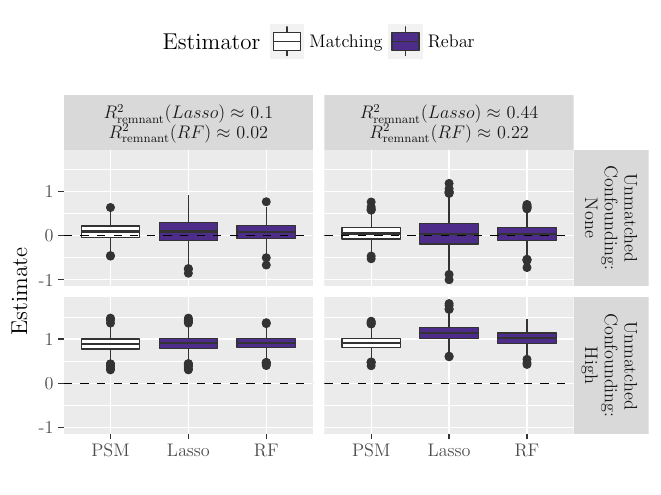} 

\end{knitrout}

\caption{Boxplots of standardized treatment effect estimates from 1000
  simulation runs under the data generating models in Section \ref{sec:sim-non-linear}. The true treatment effect, indicated by a
  horizontal dotted line, is zero. The methods are optimal pair
  matching (PSM) and rebar-adjusted optimal pair matching, with
  $y_C$ predicted using lasso or random forests (RF). The four simulation
  scenarios are arranged in a matrix, with rows for
  $\kappa=$0 and 0.5 %, denoting
  %levels of confounding after the match,
and columns for $\rho=$ 0 and 0.05. %, from left to right, denoting correlation structure between covariates.
The $R^2_{\text{remnant}}$ values listed are averages of
  prediction $R^2$ for $\algorithm$ estimated using
  cross-validation within the remnant for lasso and random forest.}
\label{fig:sim-bad}
\end{figure}

We conducted a parallel simulation study to investigate rebar's performance when the
distribution of $y_C$, conditional on $X$, differs greatly between the
remnant and the matched set.
Since it is the match that determines which subjects are in the
matched set and which are in the remnant, and the data generation
occurs prior to the match, we could not set the distribution of
$y_C$ in the remnant exactly.
Instead, we let the data generating model for $y_C$ vary with
$Pr(Z=1)$, subjects' probabilities of being treated.
To do so, we modified both the outcome model (\ref{eq:outcome}) and
the treatment model (\ref{eq:treatment}).
To select treated subjects, we chose those $2n_T$ with the highest
linear predictors, as defined in equation (\ref{eq:treatment}), and assigned
half to treatment.
That left an ``untreatable'' group of subjects with $Pr(Z=1)=0$.
For the untreatable subjects, $y_C$ was generated as in
(\ref{eq:outcome}).
For the $2n_T$ subjects with $Pr(Z=1)=0.5$, the outcomes were
generated as $\overline{\covVec\beta^*}-\covVec\beta^*+\epsilon$,
where $\beta^*$ is the concatenation of a vector of five $1$s with
$\beta$ and $\overline{\covVec\beta^*}$ is the sample average of
  all subjects' $\covVec\beta^*$.
Finally, we transformed $y_C$ to $-y_C$, so that the omitted variable
bias would be positive, as in Section~\ref{sec:sim-results}.
In this study, the relationship between $\covVec$ and $y_C$ for
subjects who could be treated was precisely the opposite of the
relationship for subjects who could not.
The worry here was that $\algorithm$ would be severely
misleading, if fit in the remnant and extrapolated to the matched set.

The simulation results suggest that this is, indeed, a concern---in
some cases.
Figure~\ref{fig:sim-bad} shows the results of rebar adjustment to
optimal pair matching using two different rebar algorithms $\algorithm$: lasso,
which depends on a linear model, and random forest, which does not.
Rebar adjustment with lasso worsened the bias and variance of the matching
estimator, slightly for lower $R^2_{\text{remnant}}$ values and considerably
for higher $R^2_{\text{remnant}}$.
On the other hand, rebar using random forests, which achieved much
lower $R^2_{\text{remnant}}$ values across the board, did little to no damage
to the matching estimator.
Apparently the matching routines were
unable, in general, to perfectly identify the treatable control subjects with
$Pr(Z=1)=0.5$, so both the remnant and the matched set contained
subjects with outcomes drawn from both outcome models.
By ignoring non-linearity, the lasso was able to fit the training
sample more closely than random forests did; but because of their
sensitivity to non-linearity, random forests extrapolated beyond the
remnant more reliably than the lasso.

In summary, under data generating models combining nonlinear responses with
limited propensity score overlap, rebar's performance depended on the
prediction algorithm.
Rebar adjustment via lasso increased the MSE
of the matching estimator, while rebar adjustment via random forest
caused little to no harm.   It is unclear whether latent nonlinearities sufficient to undercut the
lasso are to be expected in applications; proximal validation may flag
some such situations.

\section{Example Data Analysis: Evaluating Board Exam Systems}\label{sec:example}

Board Exam Systems (BES) comprise a class of similar comprehensive educational reforms.
BES are packages that a school can adopt: sets of rigorous curricula for all academic courses, corresponding sets of end-of-course exams, professional development and instructional guidance for teachers and systems of assistance for struggling students.
Though uncommon in the United States, BES are common around the world, and several research studies have suggested that they improve student achievement \citep{bishop1997effect,bishop2000curriculum, collier2009institutional}

Seven Arizona High Schools began implementing BES programs in the
2012--2013 school year: either the ACT Quality
Core program or the Cambridge program.
A pilot study sought to evaluate the results after one year, in part by
estimating the effects of the BES programs on 10th-graders' end-of-year standardized test
scores---specifically, the Arizona Instrument to Measure Standards, or AIMS.
Here we present a simplified version of the study's estimate of the effect of BES on school-average 10th-grade AIMS Reading scores.
The analysis we present here is intended to illustrate the rebar method, not to evaluate the effectiveness of BES programs in Arizona.

For Arizona high schools in our sample, we had four years of pre-treatment data.
That is, data from four cohorts of students who preceded the adoption
of BES---students set to graduate in 2011--2014.
For each cohort, we had the total enrollment, the percents of students who were male, white, black, Hispanic, other race or ethnicity, receiving free or reduced-price lunches (FRL), special education (SPED), and English language learners (ELL), in addition to average 8th Grade and 10th Grade AIMS scores on writing, reading, math and science.
We also had the percent of students in each cohort with missing AIMS English and Math scores.
From these data, we computed composite AIMS scores by averaging the
four AIMS components, and school ``trends'' for 10th grade math and reading scores: ordinary-least-squares slope estimates from the school-level regressions of school mean AIMS scores on a linear time variable.
From the US National Center for Education Statistics Common Core of
Data \citep{ccd}, we had a categorization of each school into one of 10 categories of urbanicity, ranging from urban to remote rural.
All in all, there were 90 covariates, for a total of 509 high schools.

\subsection{A Propensity Score Match}\label{sec:match}

We constructed a propensity score match to estimate treatment
effects.
Since there were only $n_{T}=7$ intervention schools, estimating
propensity scores with logistic
regression including all 90 predictors was not feasible.
%Moreover, more modern techniques for binary outcomes would also produce propensity score estimates that were either highly unstable---liable to change under small variations in the model or the data---or uninterpretable.
Instead, our propensity score model incorporated only a small subset
of the covariates, those that we believed would be most recognizable
as potential confounders to the end audience of the research.
% These included demographics and prior achievement measures, along with
% time trends.
Specifically, we regressed schools' BES status on the percent FRL, white, SPED, Hispanic, and average and percent missing 8th and 10th grade AIMS scores for students in the cohort immediately prior to BES implementation (those set to graduate in 2014) along with estimated school trends in English and Math AIMS scores.
Since this still gave more predictors than there were observations in
the treatment group, we expected that classical logistic regression
would fail to fit, so we instead used the Bayesian variant implemented in the \texttt{arm}
library for \texttt{R} \citep{arm,gelman2008weakly}.
%The resulting estimates, divided by treatment status, are displayed in Figure \ref{boxplot}

% next para quoted in full in reviewer response accompanying R3
We constructed optimal propensity-score matches, using the \verb|R|
\verb|optmatch| package \citep{hansen2007optmatch} to minimize paired
differences in the estimated log odds of assignment to treatment.
Given the relatively large pool of available comparison schools, we
disallowed the sharing of controls, as in nearest-neighbor matching or
full matching, while permitting multiple matches per treatment
schools. Rather than leaving the maximum number of matched comparisons
per treatment unspecified, we restricted it to 4, a restriction that
reduces the overall information content of the matched sample
\citep{cinarZubizarreta2016} only modestly relative to matching
without an upper limit on the number of matched controls per
treatment.  (Each matched set $\match$ makes a contribution to
effective sample size comparable to $h(n_{T\match}, n_{C\match})$
matched pairs, where
$h(n_{T\match},n_{C\match}) = \{\frac{1}{2}(n_{T\match}^{-1} + n_{C\match}^{-1})\}^{-1}$ is
the harmonic mean of $n_{T\match}$ and $n_{C\match}$
[\citealp{hansen2011propensity,cinarZubizarreta2016}].  For $n_{T\match}=1$
and $n_{C\match}\geq 1$, this contribution varies between 1 and 2, with
$h(1, 4)=1.6$.)  If this left plausible matches for some
treatment-group schools on the table, these eligible but unused
comparisons would enhance the value of proximal validation, improving
its ability to detect shortcomings of the extrapolation that underlies rebar.
%prev para quoted in full in reviewer response accompanying R1.

Table \ref{table:balance} displays covariate balance for the variables
in the propensity score model---standardized differences in covariate
means and Z-scores---before and after matching.
Covariate balance was assessed with the \verb|xBalance| routine in the
\verb|RItools| package from \verb|R| \citep{bowers2010ritools}.
The \verb|xBalance| routine also returns the results of omnibus
balance tests, for the full sample and the matched sample.
They returned p-values of 0.04
and 0.71, respectively.
Evidently, the propensity score match controlled some covariate
imbalance that was in the full sample.

\subsection{Rebar to Adjust the Match}

% latex table generated in R 3.3.1 by xtable 1.8-2 package
% Thu Sep 14 11:05:07 2017
\begin{table}[ht]
\centering
\begin{tabular}{lrrrr}
  \hline
  & \multicolumn{4}{c}{Std. Diff.} \\
 &\multicolumn{2}{c}{Unmatched}&\multicolumn{2}{c}{Matched}\\
 \hline
\% FRL & 1.06 & **  & 0.08 &     \\ 
  \% White & -0.97 & *   & 0.02 &     \\ 
  \% Sp.Ed. & -0.01 &     & -0.19 &     \\ 
  \% Hispanic & 1.34 & *** & 0.03 &     \\ 
  Urban & 0.24 &     & 0.13 &     \\ 
  avg. AIMS Writing (8th) & 0.31 &     & -0.10 &     \\ 
  avg. AIMS Reading (8th) & 0.42 &     & -0.18 &     \\ 
  avg. AIMS Math (8th) & 0.79 & *   & 0.06 &     \\ 
  avg. AIMS Reading (10th) & -0.55 &     & 0.14 &     \\ 
  avg. AIMS Math (10th) & -0.27 &     & 0.05 &     \\ 
  avg. AIMS Writing (10th) & -0.46 &     & -0.01 &     \\ 
  trend: AIMS English (10th) & -0.37 &     & 0.11 &     \\ 
  trend: AIMS Math (10th) & -0.42 &     & 0.10 &     \\ 
  \% AIMS Eng. Missing & -0.27 &     & -0.17 &     \\ 
  \% AIMS Math Missing & -0.20 &     & -0.22 &     \\ 
  $\yhat (\covVec)$ & -0.06 &     & 0.14 &     \\ 
   \hline
\end{tabular}
\caption{Standardized differences testing balance on
  covariates from the propensity score model and predictions $\hat{y}_C(\covMat)$ in the entire sample of schools and for the matched sample, conducted with the xBalance procedure.} 
\label{table:balance}
\end{table}

\subsubsection{Estimating  $\algorithm$}
% latex table generated in R 3.3.1 by xtable 1.8-2 package
% Thu Sep 14 11:05:07 2017
\begin{table}[ht]
\centering
\begin{tabular}{rrrrrr}
  \hline
 & Lasso & Random Forest & BayesLM & Ridge & Mean \\ 
  \hline
RMSE & 18.10 & 15.56 & 45.36 & 18.54 & 26.93 \\ 
  $R^2$ & 0.55 & 0.66 & -1.85 & 0.52 & -0.00 \\ 
  coefficient & 0.00 & 1.00 & 0.00 & 0.00 & 0.00 \\ 
   \hline
\end{tabular}
\caption{CV root-mean-squared error,  $R^2$, and ensemble learner weight from the SuperLearner. The seven models displayed are the lasso, random forest, a linear model with weak priors on the coefficients (``BayesLM''), ridge regression, and a grand mean model} 
\label{cvTab}
\end{table}

After setting aside the treated schools and their untreated matches, there were 483 schools in the remnant.
We considered four different predictive modeling strategies to
construct $\algorithm$: the lasso, random forests, ridge regression
\citep{hoerl1970ridge,MASS}, and linear regression with weak priors for
regularization \citep{arm},
 along with
grand-mean prediction, all combined via the
SuperLearner.
The SuperLearner uses cross validation to estimate the prediction mean-squared-error of each of the
modeling algorithms in a library.
Then, it constructs an ``ensemble learner,'' predicting new values as
a weighted average of the predictions from each of the algorithms,
with the weights determined by the cross-validation results.
These results are displayed in Table \ref{cvTab}.
The random forest dominated the other algorithms, with a
prediction $R^2$ of 0.66, to the
extent that its ensemble weight was
1.

\subsubsection{Proximal Validation}

\begin{figure}
\begin{knitrout}
\definecolor{shadecolor}{rgb}{0.969, 0.969, 0.969}\color{fgcolor}
\includegraphics[width=\maxwidth]{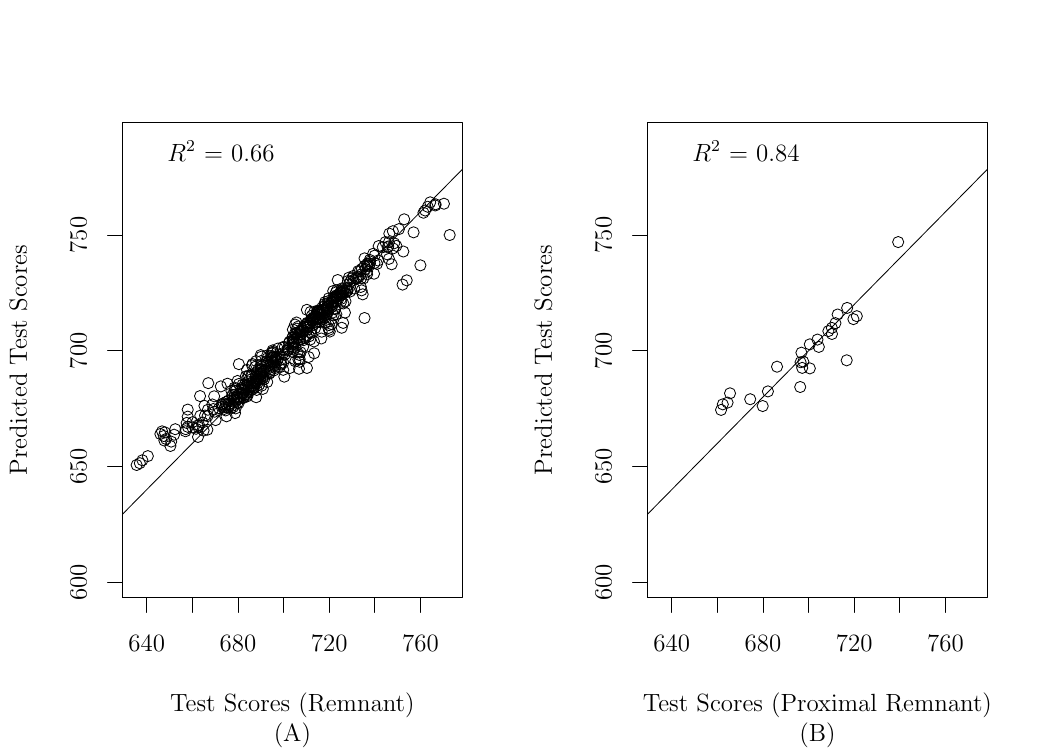} 

\end{knitrout}
\caption{SuperLearner prediction accuracy:
predictions ($\hat{y}_C(\covMat)$) as a function of real test scores. (A) gives
the results of the SuperLearner fit to, and tested against, the entire remnant. (B) shows
the proximal validation results: the performance of the SuperLearner fit in the distal portion of the remnant and tested against the proximal portion. The figures also contain the $y=x$
line for comparison.}
\label{fig:cv}
\end{figure}

To gauge how a model trained on the remnant might perform on the matched sample, we conducted proximal validation, described in Section~\ref{sec:remnant}.
First, we constructed a second match, $\match^{big}$, identical to the first, but allowing each treated subject to match at most 10 control subjects.
Match $\match^{big}$ included an additional
31 control schools in the
matched set---proximal schools---leaving 452 distal schools as a training set.
We trained the SuperLearner on the distal schools, and computed its
prediction accuracy against the proximal schools.
Somewhat surprisingly, the prediction models performed better when
trained on the distal schools and tested on the proximal schools than
when trained and tested on random subsets of the remnant in cross-validation.
This may be a result of sampling error, or the fact that the distal set contains a number of outlier schools whose AIMS reading scores are particularly hard to predict.
These schools will increase the estimated MSE reported by any validation method that includes them in its testing set.
If there are no outlier schools in the proximal set, proximal validation will not suffer from this difficulty.

As an additional check of the identification assumption (\ref{assumption:experiment}) for match $\match$, we tested balance on $\hat{y}_C(\covMat)$, in the same way as for other covariates: we tested if $\EE \bm{\hat{y}}_C^T\bm{Z}/n_T=\EE \bm{\hat{y}}_C^T(\bm{1}-\bm{Z})/n_C$.
The resulting p-value from the \verb|xBalance| routine was
0.5%$
; the balance test on $\hat{y}_C(\covMat)$ does not falsify (\ref{assumption:experiment}).

\subsubsection{Estimating Treatment Effects}

% latex table generated in R 3.3.1 by xtable 1.8-2 package
% Thu Sep 14 11:10:19 2017
\begin{table}[ht]
\centering
\begin{tabular}{rllll}
  \hline
 & Estimate & SE & p-value & 95\% CI \\ 
  \hline
PSM & 5.91 & 4.98 & 0.48 & (-10.4,22.53) \\ 
  rebar & 1.82 & 3.65 & 0.57 & (-5.41,12.17) \\ 
   \hline
\end{tabular}
\caption{The average treatment effect on the treated $\tau_{ETT}$, along with regression standard errors and permutational p-values and 95\%
confidence intervals, estimated with conventional propensity-score matching, as described in Section~\ref{sec:match}, and with rebar.} 
\label{table:results}
\end{table}

Finally, we calculated both $\tau_\Match$, the matching estimator
using $Y$, and $\est$, the rebar matching estimator; these are shown
in Table~\ref{table:results} along with HC3 standard errors.
To estimate p-values, we conducted permutation tests, permuting treatment indicators within matched sets and re-computing the estimates.
Ninety-five percent confidence intervals were estimated by inverting the permutation test, as in \citet[e.g.][]{rosenbaum2002covariance}.
Neither the conventional method nor rebar detected a statistically significant effect.
However, the rebar estimate resulted in a confidence interval with less than half the width of the conventional interval.

An anonymous reviewer suggested a post-hoc assessment of $\algorithm$'s
fit: estimating $R^2_\Match$ by comparing $y_C$ from within the
match to corresponding predictions $\hat{y}_C$.
The result was $\hat{R}^2_\Match=$0.71.

\section{Conclusion}\label{sec:conclusion}

In structural engineering, ``rebar'' abbreviates ``reinforcement
bar,'' a metal beam that is embedded in concrete. Concrete is resistant to compression, whereas rebar is resistant to tension; the combination of the two materials, rebar and concrete, is robust to a variety of threats.
Similarly, the rebar method of this paper complements the use of matching for confounder control.
Whereas matching typically focuses primarily on possible confounders'
associations with the treatment variable, and typically leaves some
subjects unmatched, rebar addresses bias by using the remnant from
matching, the unmatched controls, to model possible confounders'
associations with outcomes.
The predictions that result, $\yhat (\covVec)$, extract information about
subjects' control potential outcomes from the covariates $\covMat$.
The process of residualizing, that is, subtracting predictions
$\yhat (\covVec)$ from outcomes $Y$, can neutralize confounding from variables that the
match failed to balance.

Residualizing using the remnant confers these benefits without
compromising the statistical rationale for matching.
Indeed, matching supplemented with rebar inherits a number of central
attractions of the matching estimator.  For instance, researchers with
any level of statistical training can assess the success of the
matching procedure by examining matched units' comparability on
substantively meaningful baseline variables.  Although it typically
makes use of data from outside the range of common support---the set
of subjects $i$ for which $0<Pr(Z_i=1|\covVec_i)<1$---its final
estimate $\est$ compares only matched subjects, observing any common
support restrictions that the matching procedure observed.  The
procedure is compatible with postponing analysis involving outcomes
until the process of matching is complete, as recommended by
\citet{rubin2008objective}.  If matching succeeds in recreating a
latent experiment, where subjects matched to each other were assigned
to treatment randomly, then $\est$, like
$\hat{\tau}_\Match$, is unbiased.

Generating predictions $\yhat (\covVec)$ involves extrapolating from
the remnant to the matched sample; in some circumstances, the method could worsen the
quality of matched inferences. This risk is mitigated with the use of
cross-validation, to limit overfitting of the prediction model,
followed by proximate validation, which additionally detects biases
specific to extrapolation from lower- into higher-propensity score regions of
$\covVec$-space.  Both forms of validation are assisted by the presence of a
sizable matching remnant, including at least controls that would
have been suitable matches for some treatment group members.
While compatible with any method of matching that leaves a positive
fraction of the control reservoir unmatched, rebar is
particularly attractive in observational studies with
many more untreated than treated subjects.

We have focused on the capacity of rebar to reduce bias, but the
method may have other benefits as well.  For instance, the confidence
interval from a rebar analysis of the BES data had less than half the
width of the confidence interval from the corresponding matching
analysis.  Indeed, confidence interval widths and standard errors
generally vary inversely with the variance of the outcome.  Unless the
rebar extrapolation is sufficiently unstable as to worsen MSE ---
within the matched sample, the mean-square difference between rebar's
out-of-sample prediction and $Y$ exceeds the variance of $Y$ ---
confidence intervals based on $e$ are bound to be tighter than those
based on $Y$ alone.  In addition, studies with more stable outcomes
tend to have lower design sensitivity
\citep{rosenbaum2004design,zubizarreta2013effect}.  Barring
instability, the rebar analysis will be less sensitive to confounding
from unmeasured or unmodeled variables.  The relative stability of $e$
and $Y$ is reflected in the prediction $R^2$ of the rebar $\algorithm$
when applied to the matched set, for which cross-validation and
proximal validation can suggest a plausible range.

%\newpage

%\theendnotes

\bibliographystyle{asa}
\bibliography{rebar,morecites}

\section{Appendix: Proofs of Propositions \ref{biasBound} and \ref{prop:gammaBound}}

\subsection{The Bias of $\tau_\Match$}

\begin{lemma}\label{prop:matchBias}
In a matching design where the target of estimation is $\tau_{ETT}$, the bias of matching estimator (\ref{MatchingEstimator}) is
\begin{align*}\label{eq:matchBias}
\EE[\hat{\tau}_\Match (\bm{Y})]-\tau_{ETT}&= \sum_m w_m\bycm^T(\frac{\bm{p}_m}{\ntm}-\frac{\bm{1}-\bm{p}_m}{\ncm})
\end{align*}
where $\bycm$ is the vector of $y_C$ values for all subjects for whom $\match_i=m$: $\{\yci\}_{\match_i=m}$, and $\bm{p}_m$ is a vector of probabilities of treatment assignment for subjects in $m$, given $\ntm$ and $\ncm$: $P_i=Pr(Z_i=1|\ntm,\ncm)$.

\end{lemma}
\begin{proof}
All of the following expectations are taken conditional on $n_{C1},...,n_{CM}$ and $n_{T1},...,n_{TM}$.
\begin{align*}
	\EE\hat{\tau}_\Match&=\EE\sum_m w_m t_m(Y,Z)\\
	&=\sum_m w_m \EE t_m(Y,Z)
\end{align*}
Next, for a particular match $m$,
\begin{align*}
  \EE t_m(Y,Z) &= \EE[\frac{1}{n_T}  \bm{Y}_m^T\bm{Z}_m-\frac{1}{\ncm}\bm{Y}_m^T(\bm{1}-\bm{Z}_m)]\\
  &= \EE[\frac{1}{n_T}  \bm{y_{Cm}}^T\bm{Z}_m-\frac{1}{\ncm}\bm{y_{Cm}}^T(\bm{1}-\bm{Z}_m)]+\EE\frac{\bm{\tau}_m^T\bm{Z}_m}{\ntm}\\
  &=\bm{y_{Cm}}^T\EE[\frac{1}{\ntm}\bm{Z}_m-\frac{1}{\ncm}(\bm{1}-\bm{Z}_m)]+\frac{\bm{\tau}_m^T\EE\bm{Z}_m}{\ntm}\\
  &=\bycm^T(\frac{\bm{p}_m}{\ntm}-\frac{\bm{1}-\bm{p}_m}{\ncm})+\frac{\bm{\tau}_m^T\EE\bm{Z}_m}{\ntm}
\end{align*}
Then note that $\sum_m w_m
\frac{\bm{\tau}_m^T\EE\bm{Z}_m}{\ntm}=\tau_{ETT}$.
\end{proof}

\subsection{Proof of Proposition \ref{biasBound}}
\begin{proof}
As in Lemma \ref{prop:matchBias}, the squared bias of $\est$ is
\begin{equation*}
bias^2(\est)=\left[ \sum_m
w_m(\bm{y_{Cm}}-\bm{\hat{y}_{Cm}})^T\left(\frac{\bm{P_m}}{\ntm}-\frac{\bm{1}-\bm{P_m}}{\ncm}\right)\right]^2.
\end{equation*}
Let $\bm{y}_C$ and $\hat{\bm{y}}_C$ be length-$n$ vectors,
concatenations of $\bm{y}_{Cm}$ and $\hat{\bm{y}}_{Cm}$. For $i=1,\cdots,n$ let
$Q_i=w_{\match_i}\left(\frac{P_i}{n_{T\match_i}}-\frac{1-P_{i}}{n_{C\match_i}}\right)$
and let $\bm{Q}$ be a concatenation of $\{Q_i\}$, a length-$n$ vector.
Since $0\le P_i\le 1$, $|Q_i|\le
\max\left(\frac{1}{n_{T\match_i}},\frac{1}{n_{C\match_i}}\right)w_{\match_i}$.
Then
\begin{align*}
  bias^2(\est)&=\left[(\bm{y_C}-\bm{\hat{y_{C}}})^T\bm{Q}\right]^2\\
  &\le n\frac{||\bm{y_C}-\bm{\hat{y}_C}||^2}{n} ||Q||^2
  \text{ by Cauchy-Schwartz}\\
  &\le n\frac{||\bm{y_C}-\bm{\hat{y}_C}||^2}{n} \sum_i \max\left(\frac{1}{n_{T\match_i}},\frac{1}{n_{C\match_i}}\right)^2w_{\match_i}^2\\
  &=\frac{||\bm{y_C}-\bm{\hat{y}_C}||^2}{n} \frac{n}{n_T^2}\sum_m
  (\ncm+\ntm) \max\left(1,\frac{\ntm}{\ncm}\right)^2
  \end{align*}
\end{proof}

\subsection{Proof of Proposition \ref{prop:gammaBound}}
\begin{proof}
  The proof follows the form of the proof of
  Proposition~\ref{biasBound}, but exploits the fact that $Q_i\le
  (\Gamma^{1/2}-1)/(\Gamma^{1/2}+1)/n_T$. This follows from two facts:
  first, in a matched pair design, if $i$ is matched to $j$ and $i\ne
  j$, $P_i=1-P_j$, so (\ref{eq:rbounds}) can be re-written as
  $1/\Gamma\le P_i^2/(1-P_i)^2\le \Gamma$. Secondly, in a matched pair
  design, the term $P_i/n_{T\match_i}-P_j/n_{C\match_j}$ can be
  written as $2P_i-1$. The result follows.
\end{proof}

\end{document}